\documentclass[onecolumn,draft, 11pt]{IEEEtran}

\usepackage{color}

\usepackage{bm,graphicx,tabularx,array,geometry,amsmath,amsthm,thmtools}

\usepackage{mathtools}

\usepackage{caption}
\usepackage{amsfonts}
\usepackage{bbm}
\usepackage{makecell}
\usepackage{multirow}
 \usepackage{amssymb}
\usepackage{amsthm}
\usepackage{txfonts}
\usepackage[T1]{fontenc}
\usepackage{tikz}
\usepackage[scr=dutchcal]{mathalfa}
\let\mathscr\mathbscr
\usepackage[normalem]{ulem}

\newcolumntype{x}[1]{>{\centering\arraybackslash}p{#1}}

\declaretheoremstyle[headfont=\bfseries, 
    bodyfont=\normalfont]{normalhead}
\declaretheorem[style=normalhead]{Example}

\allowdisplaybreaks 

\usepackage[]{algorithm2e}
\newtheorem{Theorem}{Theorem}

\newtheorem{Lemma}{Lemma}

\newtheorem{Remark}{Remark}
\newtheorem{Definition}{Definition}

\begin{document}

\title{ An Information Theoretic Framework for Active De-anonymization in Social Networks Based on Group Memberships}


\author{Farhad Shirani, Siddharth Garg, and Elza Erkip
\\
\{fsc265,siddharth.garg,elza\}@nyu.edu\\
Department of Electrical and Computer Engineering\\
New York University, New York, New York, 11201 \\\date{} }

\maketitle

\begin{abstract}
In this paper, a new mathematical formulation for the problem of de-anonymizing social network users by actively querying their membership in social network groups is introduced. In this formulation, the attacker has access to a noisy observation of the group membership of each user in the social network. When an unidentified victim visits a malicious website, the attacker uses browser history sniffing to make queries regarding the victim's social media activity. Particularly, it can make polar queries regarding the victim's group memberships and the victim's identity. The attacker receives noisy responses to her queries. The goal is to de-anonymize the victim with the minimum number of queries. Starting with a rigorous mathematical model for this active de-anonymization problem, an upper bound on the attacker's expected query cost is derived, and new attack algorithms are proposed which achieve this bound. These algorithms vary in computational cost and performance. The results suggest that prior heuristic approaches to this problem provide sub-optimal solutions.  
\end{abstract}


%
\IEEEpeerreviewmaketitle

\section{Introduction}
With the rapid development of information technology, the internet has become an integral part of our daily lives. User anonymity is arguably one of the most critical aspects of the cyberspace. Not only is anonymity a preference of the users, but it is also necessary to ensure data security and to shield users from data-stealing attacks. 

Social networks such as Facebook, Twitter, and LinkedIn and the social data which they encompass find applications in academic research, government, business, and healthcare \cite{Beyah}.
As a result of the massive amount of personal data stored in these networks, de-anonymization attacks using social networks has been a subject of great interest in recent years \cite{Beyah,kruegel,ref1,ref2,ref3}. 


 The \emph{active} de-anonymization problem - first proposed in \cite{kruegel} - considers a scenario where a victim visits a malicious website controlled by the attacker. The attacker wishes to find the victim's real-world identity. In \cite{kruegel}, it is shown that by scanning the social network and then using browser history sniffing, the attacker can infer the victim's social media activity. Particularly, the attacker uses two sources of information for de-anonymization: 
 \\1) Prior partial information regarding the network graph, which is a bipartite graph containing information regarding the user's group memberships. The graph consists of a i) a set of user nodes, ii) a set of group nodes, and iii) a set of edges connecting users to the groups in which they are members. 
 \\2) Responses to online queries sent to the victim's device. These queries can be divided into two categories:
\begin{itemize}
\item{User identity (UID): A UID query asks if the unknown victim $u_J$ is user $u_i$ in the network. For instance, if the attacker is using browser history sniffing, then the query would return a ``yes" if the victim signed into the social network under the specific account. In general, the response could have both false positives and negatives. }
\item{group Membership Queries (GM): A GM asks if the unknown victim $u_J$ is a member of the group $r_j$ .  For instance, in browser history sniffing, the attacker verifies whether the victim visited the group's URL or not.}
\end{itemize}
  The attack strategy in \cite{kruegel} is to sweep over all groups in the social networks and to create a group membership signature for the victim. Then, the attacker intersects the members of groups which the victim is a member of, hence greatly reducing the number of candidates for de-anonymization. The attacker then sends UID queries regarding every member of the intersection until it finds a match in the victim's browser history at which point the attack terminates. Although the effectiveness of the attack on real-world social networks has been illustrated in \cite{kruegel}, given the lack of rigorous analysis, it is not clear whether and how the performance of the attack compares with the optimal performance. The first contribution of this work is to construct a mathematical framework for de-anonymization attacks. In our model, the group memberships in the social network are modeled by a random bipartite graph. This graph is called the \textit{social network graph}. An example of a social network graph is illustrated in Figure \ref{Fig:bipart}. It is assumed that the victim's index is chosen randomly and uniformly from the set of user indices $[1,m]$. We note that the active attack in \cite{kruegel} does not exploit friendship relationship's between users  and thus we do not

\begin{figure}[t]
\includegraphics*[draft=false,scale=0.5]{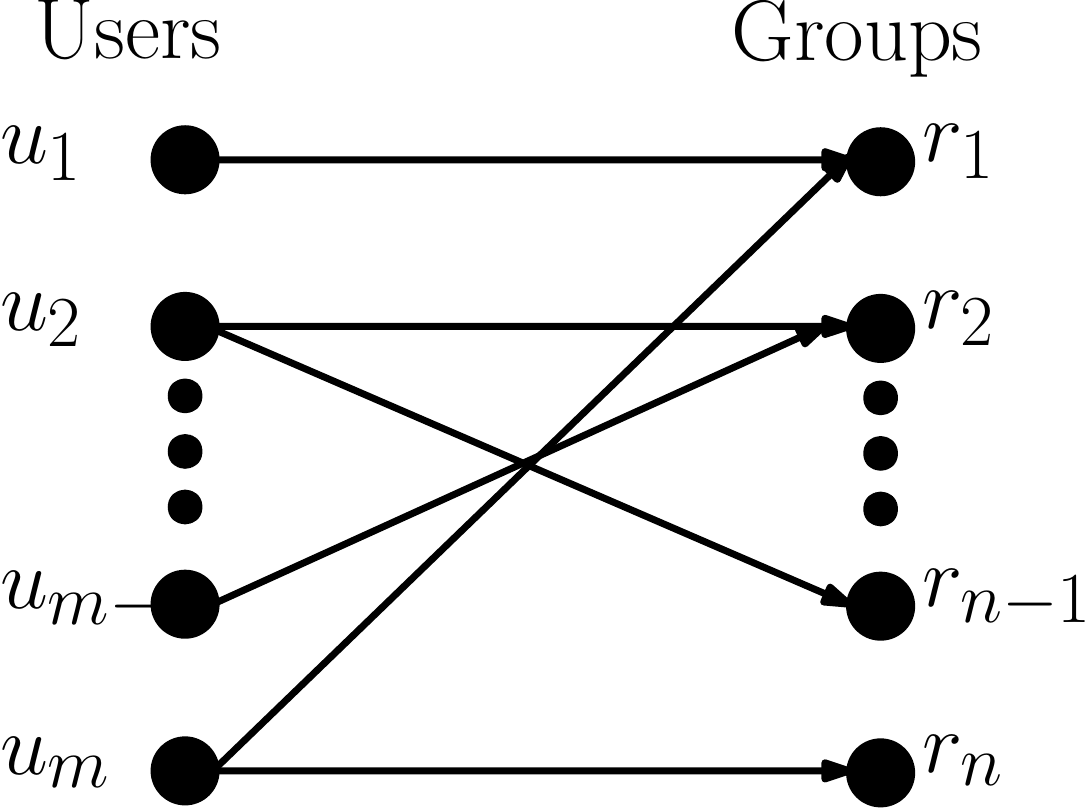}
\centering 
\caption{\textit{The figure illustrates the social network graph. The graph consists of two sets of nodes: i) the user nodes $\{u_1,u_2,\cdots,u_m\}$, and ii) the group nodes $\{r_1,r_2,\cdots,r_n\}$. An edge between user $u_i$ and group $r_j$ signifies the membership of $u_i$ in $r_j$.}}
\label{Fig:bipart}
\end{figure}

\noindent model these relationships in the social network graph.
 The attacker wants to reveal the victim's index, and has access to a noisy version of the social network graph. This models the attacker's knowledge of the social network prior to the active phase of the attack (i.e. before the victim visits the attacker's website). During the active phase of the attack, the attacker makes queries regarding the victim's group memberships and receives noisy responses to these queries. The noise models private groups in the network which are inaccessible to the attacker as well as limited access to user history due to browser security or due to limited cache space. The goal is to design algorithms for reliable de-anonymization which require the minimum expected number of queries possible.  

Our second contribution is the design and analysis of several new de-anonymization strategies which leads to upper bounds on the expected number queries required for de-anonymization. Specifically, we show that it is possible to have the expected total number of queries and the expected number of group membership queries both grow logarithmically in the number of users where the expectation is taken over all social network graphs and user indices. Hence, the attack strategy in \cite{kruegel}, in which the attacker queries every group in the network is sub-optimal if the number of groups grows super-logarithmically in the number of users. 

The rest of the paper is organized as follows: in Section II, we explain the problem formulation. In Section \ref{sec:schemes}, we investigate the new attack strategies for the active de-anonymization problem. We find upper bound to the expected number of queries resulting from each of these strategies.  
Finally, Section V concludes the paper. 
\section{Problem Formulation}
In this section, we introduce the notation and problem formulation. We model the group memberships in the social network by a bipartite graph $g^0$.  
\begin{Definition}
A bipartite graph $g^0$ consists of a triple of sets $(\mathcal{U}^0,\mathcal{R}^0,\bm{\mathcal{E}^0})$. The vertices in the graph are partitioned into two sets: 1) the user set $\mathcal{U}^0=\{u_1,u_2,\cdots,u_m\}$, and 2) the group set $\mathcal{R}^0=\{r_1,r_2,\cdots,r_n\}$. The set of edges is denoted by $\bm{\mathcal{E}^0}\subset\{(i,j)|i\in [1,m], j\in[1,n]\}$. 
\end{Definition}
If $(i,j)\in \bm{\mathcal{E}^0}$, we say that `user $u_i$ is in group $r_j$'. The set of members of the group $r_j$ is denoted by $\bm{\mathcal{E}^0}_j=\{i| (i,j)\in \mathcal{E}^0\}, j\in [1,n]$.  The set of groups associated with user $i$ is defined as $\mathcal{F}^0_i=\{j|(i,j)\in \mathcal{E}^0\}, i\in [1,m]$. The group signature of user $i$ is the vector $\underline{F}_i= (F_{i,1},F_{i,2},\cdots,F_{i,n})$, where
\begin{align*}
F_{i,k}=
\begin{cases}
 1 \qquad &\text{if } u_i\in \mathcal{E}^0_k,\\
0 &\text{otherwise}.
\end{cases}
\end{align*}
That is, the group signature is the n-length binary vector for which the $j$th element is 1 if and only if $u_i$ is a member of the group $r_j$. The vector $F_{j,n_1}^{n_2}=(F_{j,n_1},F_{j,n_1+1},\cdots, F_{j,n_2})$ is called a partial group signature of the $j$th user, where $1\leq n_1<n_2\leq n$.
  

\begin{figure}[t]
\includegraphics*[draft=false,scale=.8]{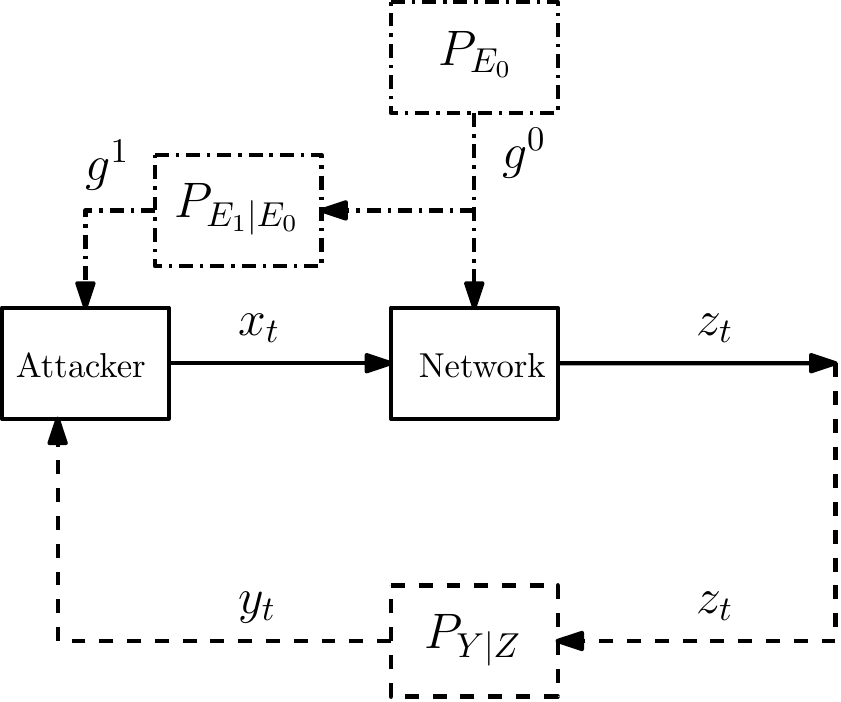}
\centering 
\caption{\textit{The figure illustrates the components of the active de-anonymization problem. {Initially, the dashed-dotted components generate the graph pair $(g^0,g^1)$}. At time $t$, the attacker sends the query $X_t$ to the network. $X_t$ is a function of $(g^1,Y^{t-1})$.
Then, the network generates $Z_t$ to the attacker, where $Z_t$ is a function of $g^0$. $Z_t$ is corrupted by noise and the attacker receives the noisy output $Y_t$.}}
\label{Fig:tree}
\end{figure}
\begin{Example}
 In the Facebook social network, $\mathcal{U}^0$ is the set of users and $\mathcal{R}^0$ is the group set which includes the pages/ events/ groups/  applications on the social network. An edge between user $i$ and group $j$ indicates that user $i$ is a member of group $j$. Here, the groups of interest are those whose member lists are publicly available. 
\end{Example}
For simplicity, we assume that the set of edges $\bm{\mathcal{E}^0}$ is generated randomly with edge probability $p$. More precisely, let $E_0$ be a Bernoulli random variable with $P_{E_0}(1)=p$. Then, the probability that user $i$ is in group $j$ is $P((i,j)\in \bm{\mathcal{E}^0})=P_{E_0}(1)=p$ and is independent of all other group memberships. So:
\begin{align*}
 P(\bm{\mathcal{E}^0}=\mathcal{E}^0)=\!\!\!\!\!\!\!\!\!\!\prod_{(i,j)\in [1,m]\times[1,n]} \!\!\!\!\!\!\!\!\!\!P_{E_0}\left(\mathbbm{1}\left((i,j)\in \mathcal{E}^0\right)\right).
\end{align*}
Generally, it is assumed that the attacker has access to a noisy version of the graph $g^0$ which we denote by $g^1$ \cite{ped}. The noise is assumed to have a specific `single-edge' structure. The presence of any edge $(i,j)$  in the graph $g^1$ depends only on its presence in the graph $g^0$ and is independent of all other edges. Formally, let $E_1$ be a binary random variable with the following distribution conditioned on $E_0$:
\begin{align}
 P_{E_1|E_0}(\alpha|\beta)=
 \begin{cases}
 1-e_1 \qquad & \text{if } \alpha=1,\beta=1,\\
 e_1 & \text{if } \alpha=0, \beta=1,\\
 e_2& \text{if } \alpha=1, \beta=0,\\
 1-e_2& \text{if } \alpha=0,\beta=0,\\
\end{cases}
\label{eq:noisy}
\end{align}
where $e_1,e_2\in [0,1]$. The attacker has access to the graph $g^1$ which is characterized by  $(\mathcal{U}^0,\mathcal{R}^0,\bm{\mathcal{E}^1})$ where
\begin{align*}
Pr(\bm{\mathcal{E}^1}=\mathcal{E}^1|\bm{\mathcal{E}^0}=\mathcal{E}^0)=\!\!\!\!\!\!\!\!\!\!\prod_{(i,j)\in [1,m]\times[1,n]} \!\!\!\!\!\!\!\!\!\!P_{E_1|E_0}\left(\mathbbm{1}\left((i,j)\in\mathcal{E}_1\right)|\mathbbm{1}\left((i,j)\in\mathcal{E}_0\right)\right).
\end{align*}
The group signature of user $i$ in $g^1$ is denoted by $\widehat{\underline{F}}_{i}$. The partial group signature corresponding to the first $n'\in \mathbb{N}$ groups is denoted by $\widehat{{F}}_{i,1}^{n'}$.
\begin{Example}
One way for the attacker to discover the bipartite graph in the Facebook social network is by scanning the network \cite{kruegel}. In practice, the attacker observes a subset of the true group memberships/interests of users. For instance some users choose to keep their membership in certain groups private. In other words, the attacker samples the graph $g_0$.   
The sampled graph $g^1$ is characterized by  $(\mathcal{U}^0,\mathcal{R}^0,\bm{\mathcal{E}^1})$ where
\begin{align}
Pr(\bm{\mathcal{E}^1}=\mathcal{E}^1|\bm{\mathcal{E}^0}=\mathcal{E}^0)=\mathbbm{1}(\mathcal{E}^1\subset \bm{\mathcal{E}^0})s_1^{|\bm{\mathcal{E}^1}|}(1-s_1)^{|\bm{\mathcal{E}^0}- \bm{\mathcal{E}^1}|}. 
\label{eq:fb}
\end{align}
Comparing Equations \eqref{eq:noisy} and \eqref{eq:fb}, it follows that in this instance of the problem, $e_1=s_1$ and $e_2=0$. This can be interpreted as follows: Each edge $(i,j)$ in $\bm{\mathcal{E}^0}$ is included in $\bm{\mathcal{E}^1}$ with probability $s_1$ and it is omitted with probability $1-s_1$. The sampled graph represents the attacker's `knowledge' of the graph $g^0$. 
\end{Example}
It is assumed that a user index $J$ is chosen randomly based on  a probability distribution $P_J$ on the set of user indices $[1,m]$. In this work, we assume that $P_J$ is the uniform distribution.
 The attacker's goal is to determine the realization of $J$. The attacker may make two forms of queries, 1) UID: Is $J$ equal to $j$ for some $ j\in [1,m]$?, and 2) GM: Is user $u_J$ in group $r_i$ for some  $i\in [1,n]$?
 . Equivalently, the attacker may inquire about the value of $F_{J,i}, i\in [1,n]$, or it can inquire about the value of the indicator function $\mathbbm{1}(J=j)$. After the $t$'th query, the attacker receives a binary response $z_t$, where 1 indicates a `yes' and 0 indicates a `no' response.

Formally, an attack strategy is defined below:
\begin{Definition}
 An attack strategy for the pair of graphs $(g^0,g^1)$ is defined as a sequence of functions $x_t:\{0,1\}^{(t-1)}\to\mathcal{R}^0\cup \mathcal{U}^0, t\in \mathbb{N}$. 
\end{Definition}
In the above definition, at the $t$th step the attack strategy provides a mapping $x_t$ from the vector of prior responses received by the attacker $y_1^{(t-1)}$ to the set of all possible queries. We often write $x_t$ to represent $x_t(y_1,y_2,\cdots,y_{t-1})$ which is the $t$th query. Furthermore, we call $x_t$ a group query if $x_t\in \mathcal{U}^0$, otherwise if $x_t\in \mathcal{R}^0$, we call $x_t$ a UID query.

At step $t$, the attacker transmits a symbol $x_t$ from the set $\mathcal{R}^0\bigcup \mathcal{U}^0$. The network outputs a response $z_t$ as defined below:
 \begin{align*}
 z_t=
 \begin{cases}
 1\qquad &\text{if } x_t=r_i, J\in \mathcal{E}^0_{r_i} \text{ or }x_t=u_J,\\
 0 \qquad &\text{Otherwise.}
\end{cases}
\end{align*}
In general, the attacker receives a noisy response to its queries. In this work, we assume that the responses to group membership queries are corrupted by noise, whereas for UID queries, the response is received noiselessly. This assumption simplifies the error analysis for the proposed schemes. As a result, if $x_t\in \mathcal{U}^0$, then the attacker receives $y_t$, where $y_t$ is the output of a binary channel characterized by $P_{Y|Z}$ which is given below:
\begin{align}
 P_{Y|Z}(\alpha|\beta)=
 \begin{cases}
1-f_1, \qquad & \text{if } \alpha=1,\beta=1,\\
 f_1, & \text{if } \alpha=0, \beta=1,\\
 f_2,& \text{if } \alpha=1, \beta=0,\\
 1-f_2,& \text{if } \alpha=0,\beta=0,\\
\end{cases}
\label{eq:noisyout}
\end{align}

where $f_1,f_2\in [0,1].$
\begin{Example}
In \cite{kruegel}, it is assumed that the attacker uses the user's cached history to find the answers to its queries. Since the cached data is limited, some queries might result in false negatives. In this formulation, $y_t$ is the output of a
binary $z$-channel with parameter $s_2\in [0,1]$ and input $z_t$.
In the form of Equation \eqref{eq:noisyout}, $f_1=s_2$ and $f_2=0$.
\end{Example}

The following defines  a measure on the performance of an attack strategy:
\begin{Definition}
 For an attack strategy characterized by $x_{t}(y_1^{t-1}), t\in\mathbb{N}$, we define the number of queries $Q\triangleq min\{t| (x_t,y_t)= (u_J,1)\}$.
 \end{Definition}
 Here, $Q$ indicates the number of queries made by the attacker before the user is de-anonymized.
For a given attack strategy, the expected number of queries is said to be $O(g(m))$ if
 \begin{align*}
\limsup_{m\to \infty}\frac{\mathbb{E}(Q)}{g(m)} \leq \infty,
\end{align*}
where the expectation is taken over the user indices $J$, and over the sets of edges $(\mathcal{E}^0,\mathcal{E}^1)$. 
\begin{Definition}
 The minimum expected number of queries is defined as
 \[\bar{Q}\triangleq\min_{x_t:\{0,1\}^{(t-1)}\to\mathcal{R}^0\cup \mathcal{U}^0, t\in \mathbb{N}}\mathbb{E}(Q)\]. 
\end{Definition}
We are interested in finding upper bounds to the value of $\bar{Q}$ as a function of the number of groups $n$, the number of users $m$, edge probabilities $p$, and noise transition probabilities $P_{E_1|E_0}$ and $P_{Y|Z}$. 
\section{Expected Number of Queries: Bounds and Achievable Strategies}
\label{sec:schemes}
In this section, we provide constructive upper bounds on the minimum expected number of queries, $\bar{Q}$. We describe three different strategies and derive an upper bound based on each of them. The first two strategies are analyzed for the noiseless case (i.e. $e_i=f_i=0, i\in\{1,2\}$.). The first strategy, called the group intersection strategy (GIS), builds upon the work in \cite{kruegel}. The second strategy, called the maximum posteriori probabilities (MAP) strategy, has lower expected number of queries than the GIS. The third strategy is called the typical set strategy (TSS) uses typical set ideas from information theory \cite{csiszarbook}. We analyze the performance of this strategy for general $e_i$ and $f_i$.

\subsection{Group Intersection Strategy (GIS)}
The group attack strategy (GIS) is  based on \cite{kruegel}.  
In \cite{kruegel}, the attacker first sends all possible GM queries, and intersects the groups for which the query responses were positive, then it sends UID queries for each member of the intersection to de-anonymize the user. In contrast, in the GIS, the attacker first chooses a subset of groups in the network. It then finds the user's partial group signature corresponding to the groups in this set by sending GM queries. 
  In the second step, similar to \cite{kruegel}, the attacker intersects the groups in this subset for which the query responses were positive. Similar to \cite{kruegel}, The attacker conducts a brute-force search over this smaller user set by sending UID queries. 
 

 We proceed by formally characterizing the algorithm. 
 The attack consists of two steps: 1) the GM query step, and 2) the UID query step. The following describes the the first step. The attacker has access to 
$g^{0}$ which is characterized by $(\mathcal{U}^{0},\mathcal{R}^{0},\bm{\mathcal{E}^{0}})$. Initially, a subset of size $n'$ of the groups is selected uniformly at random from the set $\mathcal{R}^{0}$, where $n'\in [1,n]$. Let the indices of the sets in the collection be denoted by $\{i_1,i_2,\cdots, i_{n'}\}$. 
   Then, the attacker finds the partial group signature corresponding to this set. This is done by transmitting $x_t=r_{i_t}, t=1,2,\cdots, n'$. Once the transmission is complete, the attacker reduces the graph by eliminating all users which are not a member of the intersection $\bigcap_{i_t:z_t=1}\mathcal{E}^{0}_{i_t}$. The new graph $g^1$ is characterized by $(\mathcal{U}^1,\mathcal{R}^1,\bm{\mathcal{E}^1})$, where 
\begin{align*}
&\mathcal{U}^1=\bigcap_{i_t:z_t=1}\mathcal{E}^{0}_{i_t},\\ 
&\mathcal{R}^1=\{r_j| \exists u_i\in \mathcal{U}^1, u_i\in \mathcal{E}^{0}_j\},\\
&\mathcal{E}^1=\{(j,k)| (j,k)\in \mathcal{E}^{0}, j\in \mathcal{U}^1\}.
 \end{align*}
  Let $\mathcal{U}^1=\{u_{i_1},u_{i_2},\cdots,u_{i_l} \}$. In the second step, which consists of the UID queries, the attacker transmits $x_{n'+j}= r_{n+i_j}, j\in [1,l]$ until it receives $z_{n'+j}=1$. It announces the index $i_j$ as the user's index. The GIS algorithm is given in Figure \ref{alg:GIS}.

%
%

\begin{figure}[t]
\includegraphics*[draft=false,scale=.5]{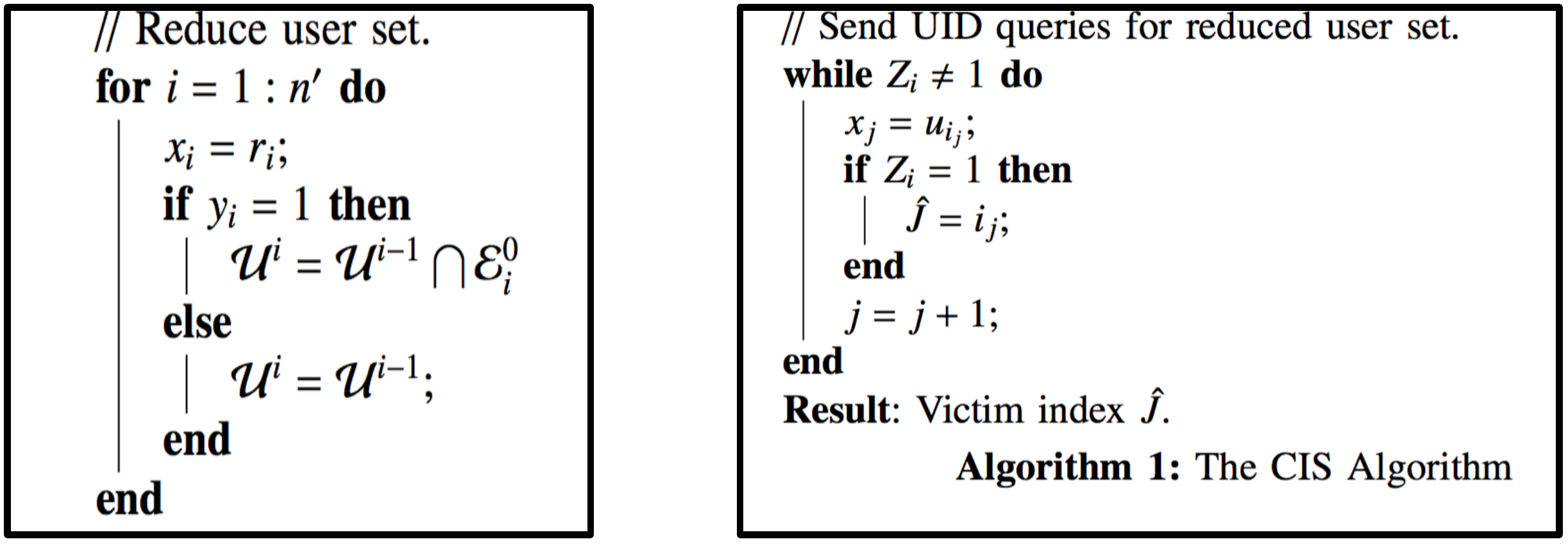}
\centering 

 \caption{\textit{The GIS Algorithm}}
 \label{alg:GIS}
\end{figure}


 The number of queries made in this attack is denoted by $Q_{GIS}$. We proceed to analyze the performance of the strategy when $e_i=f_i=0, i\in \{1,2\}$. That is, we analyze the performance assuming that the attacker has full knowledge of the network, and the responses to the queries are received noiselessly, hence $z_t=y_t$ and $g^0=g^1$. 
  
The following theorem provides bounds on the expected number of queries required in the GIS algorithm.

\begin{Theorem}
 For the $GIS$ strategy, we have:
 \begin{align}
\mathbb{E}(Q_{GIS})= n'+cm(1-p+p^2)^{n'}+O(1),
\label{eq:th21}
\end{align}
where $c=O(\sqrt{n'})$ and $n'\leq n$.
If
  \begin{align}
  \label{eq:Sp}n'=\left(\frac{1}{p(1-p)}+\frac{1}{\log_2{\frac{1}{1-p}}}\right)\log_2{m},
  \end{align}
then, $\mathbb{E}(Q_{GIS})=\left(\frac{1}{p(1-p)}+\frac{1}{\log_2{\frac{1}{1-p}}}\right)\log_2{m}+O(1)$ holds, 
which implies
\begin{align*}
\limsup_{m\to\infty}{\Bigg|\mathbb{E}(Q_{GIS})-\left(\frac{1}{p(1-p)}+\frac{1}{\log_2{\frac{1}{1-p}}}\right)\log_2{m}\Bigg|}<\infty.
\end{align*}
\label{thm:2}
\end{Theorem}
\begin{proof}
Please refer to the Appendix. 
\end{proof}
 

\subsection{Maximum a Posteriori (MAP)} 

In the maximum a posteriori probabilities (MAP) attack strategy, the attacker first determines a partial group signature of the user $u_J$ based on a subset of $n'$ groups which are chosen randomly and uniformly from the set of all groups, where $n'\leq n$, by sending $n'$ group queries and then uses the maximum likelihood posterior probabilities to find the index $J$ by sending additional UID queries. 

The MAP attack strategy consists of two steps. In the first step, the attacker makes queries regarding the values of $F_{J,k}, k\in [1,n']$ (or, equivalently, any $n'$ randomly picked groups). That is, it transmits the first $n'$ symbols in $\mathcal{R}^0$ one by one. It receives the partial group signature $F_{J,1}^{n'}$. Here, $n'$ is assumed to be arbitrary. We show that for asymptotically large $n$, if $n'$ is higher than a specific threshold, the attacker de-anonymizes the user's identity without any further queries with high probability. Whereas, if $n'$ is lower than the threshold, further UID queries are required. 

 In the second step, the attacker calculates the posteriori probabilities for each of the users. 
 Let \[l_i=Pr(J=i|z_1,z_2,\cdots, z_{n'}),i\in [1,m].\] Let $l_{(i_1)}\geq l_{(i_2)}\geq \cdots\geq l_{(i_m)}$ be the order statistics corresponding to the vector of random variables $(l_1,l_2,\cdots,l_m)$. At the $(n'+t)$th step, the attacker makes a UID query on whether the user index is equal to $i_{(t)}$.  
In other words, it transmits $x_{n'+t}=r_{n+(i_{t})}, t\in[1,m]$. The algorithms ends when the attacker receives the output $z_{n'+t}=1$. In summary, the attack strategy can be characterized as:
 \begin{align*}
 x_t= 
 \begin{cases}
 r_t\qquad &\text{ if } t\in [1,n'],\\
 u_{l_{(i_t)-n'}}& \text{ otherwise}.
\end{cases}
\end{align*}
%
%
%

\begin{figure}[t]
\includegraphics*[draft=false,scale=.5]{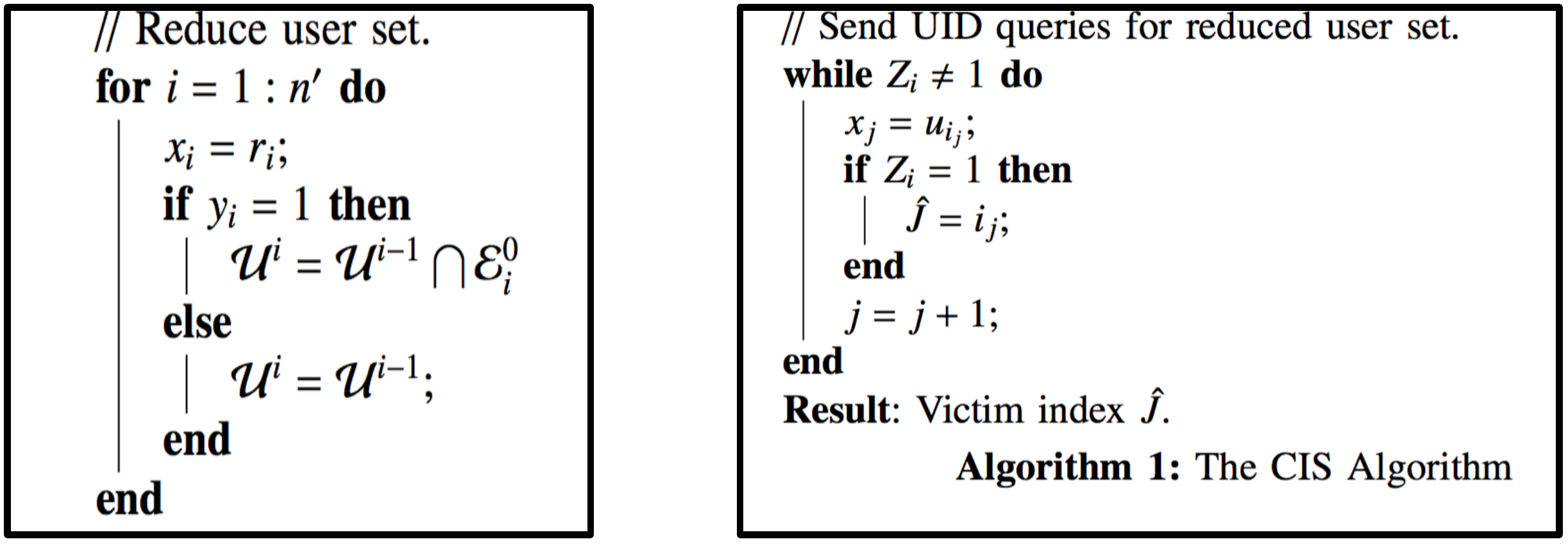}
\centering 
 \caption{\textit{The MAP Algorithm}}
 \label{alg:MAP}
\end{figure}

We analyze the performance of the strategy when $e_i=f_i=0, i\in \{1,2\}$. 
 The MAP algorithm for the noiseless de-anonymization problem simplifies to the one given in Figure \ref{alg:MAP}. In the noiseless regime, after receiving the partial group signature the users are divided into two subsets. Users who have the same partial signature as the one that is received are equally likely with a non-zero probability, whereas users which have a different partial signature have zero probability. Hence, the noiseless MAP is an extension of GIS where the attacker uses all of the entries of the response vector instead of only using positive responses.
The number of queries in the MAP strategy is denoted by $Q_{MAP}$. The following theorem provides bounds on the expected number of queries for this attack strategy when $e_i$ and $f_i$ are equal to 0. 
\begin{Theorem}
For the MAP strategy, if $e_i=f_i=0, i\in \{1,2\}$, then:
\begin{align}
\mathbb{E}(Q_{MAP})= n'+\frac{cm}{2}(1+2p^2-2p)^{n'}+O(1),
\label{eq:th11}
\end{align}
where $c=O({\sqrt{n'}})$ and $n'\leq n$.  
\\If $n'= \frac{1}{\lambda}\log_2{m}$, where $\lambda>-\log_2{(p^2+(1-p)^2)}$, then $\mathbb{E}(Q_{MAP})\leq \frac{1}{\lambda}\log_2{m}+O(\sqrt{\log_2{m}})$
.This implies, 
\begin{align}
\limsup_{m\to\infty}\frac{\big|\mathbb{E}(Q_{MAP})-\frac{1}{\lambda}\log_2{m}\big|}{\sqrt{\log_2{m}}}<\infty.
\label{eq:th12}
\end{align}
\label{thm:1}
\end{Theorem}
\begin{proof}
 Please refer to the appendix.
\end{proof}

\begin{Remark}
 Note that taking $n'= \gamma \log_2{m}, \gamma >\frac{1}{\lambda}$ in the proof of Theorem \ref{thm:1} gives $\frac{cm}{2}(p^2+(1-p)^2)^{n'}\to 0 $ as $m\to \infty$. 
Hence, for this choice of $n'$, as the number of users (and groups) increases, the number of queries after finding the group signature is negligible. So, the number of queries converges to $n'$. The probability that the resulting partial group signature is unique approaches one as $m\to \infty$. This follows by the Markov inequality:
 \begin{align*}
 P(\nexists! i: \underline{F}'=F_{i,1}^{n'})=P(|\mathcal{L}|>1)\leq \mathbb{E}(|\mathcal{L}|)\to 0,
\end{align*}
where $\mathcal{L}$ is the set of users with non-zero probability after receiving the responses to the GM queries. 
 \label{rem:end}
\end{Remark}

\subsection{Typical Set Strategy (TSS)}

The TSS is an attack strategy which is based on typical sets. We provide bounds on the performance of the TSS for arbitrary $e_i$ and $f_i$. We prove that in this strategy, the expected number of queries grows logarithmically with respect to the number of users for any fixed $e_i$ and $f_i$. We also show that the TSS strategy outperforms both the MAP strategy as well as the GIS in terms of expected number of queries when $e_i=f_i=0$.

We use the standard definitions for typical and conditional typical sets \cite{csiszarbook}. For completeness the definitions are given below.
\begin{Definition}
For a random variable $X$ defined on the probability space $(\mathcal{X},2^{\mathcal{X}},P_X)$, where $\mathcal{X}$ is a finite set, and $n\in\mathbb{N}$, and $\epsilon>0$, the typical set is defined as 
\begin{align*}
A_{\epsilon}^n(X)=\big\{x^n\big| \frac{1}{n}|N(a|x^n)-P_X(a)|\leq \epsilon, \forall a\in \mathcal{X}\big\},
\end{align*}
 where $N(a|x^n)=\sum_{i=1}^n\mathbbm{1}(x_i=a)$.
\end{Definition}
The conditional typical set is defined below.
 \begin{Definition}
For random variables $X,Y$ defined on finite sets $\mathcal{X}$ and $\mathcal{Y}$, respectively,  and for $n\in\mathbb{N}$, and $\epsilon>0$, the conditional typical set of $Y$ given a sequence $x^n$ is defined as 
\begin{align*}
A_{\epsilon}^n(Y|x^n)=\Big\{y^n\big| \frac{1}{n}|N(a,b|x^n,y^n)\!-\!P_{X,Y}(a,b)|\!\leq \epsilon, \forall a,b\in \mathcal{X}\!\times\!\mathcal{Y}\Big\},  
\end{align*}
 where $N(a,b|x^n,y^n)=\sum_{i=1}^n\mathbbm{1}(x_i=a,y_i=b)$.
\end{Definition}
We will make use of the following well-known results.
\begin{Lemma}
 Let $X_i, i\in\{1,2\}$ be random variables defined on probability spaces $(\mathcal{X}_i,2^{\mathcal{X}_i},P_{X_i}), i\in \{1,2\}$, respectively. There exists a sequence $\epsilon_k\to 0$ depending only on the cardinalities $|\mathcal{X}_i|$ so that for every distribution $P$ on $X_1$ and stochastic matrix $W:X_1\to X_2$, 
 \begin{align*}
 P_{X_1^n}(A_{\epsilon}^n(X_1))\geq 1-\frac{|\mathcal{X}_1|}{4n\epsilon^2}, \qquad P(A_{\epsilon}^n(X_2)|x^n_1)\geq 1-\frac{|\mathcal{X}_1||\mathcal{X}_2|}{4n\epsilon^2},
\end{align*}
for every $\epsilon>0$.
\label{lem:typ1}
\end{Lemma}
\begin{Lemma}
  Let $X_i, i\in\{1,2\}$ be defined as in Lemma \ref{lem:typ1}. There exists a sequence $\epsilon_k\to 0$ depending only on the cardinalities $|\mathcal{X}_i|$ so that for every distribution $P$ on $X_1$ and stochastic matrix $W:X_1\to X_2$, 
   \begin{align*}
&\big|\frac{1}{n}\log{|A_{\epsilon}^n(X_i)|}-H(P_{X_i})\big|\leq \epsilon_n,
\\&\big|\frac{1}{n}\log{P_{X_i^n}(x^n)}+H(P_{X_i})\big|\leq \epsilon_n, \forall x^n\in A_{\epsilon}^n(X_i).
\end{align*}
\label{lem:typ2}
\end{Lemma}
We define three binary random variables, $U$, $Y$, and $Z$, where the random variable $Z$ corresponds to a specific entry in the adjacency matrix of $g^0$, $U$ corresponds to the entry in the adjacency matrix of $g^1$, and $Y$ corresponds to the noisy response received after the query regarding that specific entry (i.e. $Y$ is a noisy version of $Z$.). More precisely, 
 $Z$ is distributed according to the Bernoulli distribution with $P(Z=1)=p$. The distribution of the random variable $U$ given $Z$ is $P_{U|Z}(\alpha|\beta)=P_{E_1|E_0}(\alpha|\beta)$.
The distribution of $Y$ given $Z$ is equal to $P_{Y|Z}$ given in Equation \eqref{eq:noisyout}.
 Define the joint distribution among the random variables $(U,Y,Z)$ by 
 \[P_{U,Y,Z}=P_{Y}P_{Z|Y}P_{U|Z}.
 \]
So, the Markov chain $U - Y- Z$ holds.

We proceed to describe the strategy. We fix $\epsilon>0$ and $l,n'\in \mathbb{N}$. The attack involves $l$ steps. The first step is similar to the previous two strategies. The attacker sends group membership queries corresponding to the first $n'$ groups. Hence, it makes queries regarding the values of $F_{J,k}, k\in [1,n']$. Note that in analyzing TSS, in contrast with the previous sections, it is assumed that the attacker does not necessarily have direct access to $g^0$. Rather, it has access to the noisy graph $g^1$. We denote the partial group signature of user $J$ in $g^1$ by $\widehat{F}_{J,k}^{n'}$. The attacker receives a noisy version of the partial group signature $F_{J,1}^{n'}$. We denote this noisy vector by $\widetilde{F}_{J,1}^{n'}$. The relationship between $\tilde{F}_{J,1}^{n'}$, $\widehat{F}_{J,1}^{n'}$, and $F_{J,1}^{n'}$ is shown in Figure \ref{Fig:TSS}.

The attacker verifies that $\widetilde{F}_{J,1}^{n'}\in A_{\epsilon}^{n'}(Y)$. If $\widetilde{F}_{J,1}^{n'}\notin A_{\epsilon}^{n'}(Y)$ the attacker proceeds to step two. 
Otherwise, it defines the following ambiguity set: 
\[ 
\mathcal{L}=\{i| \widehat{F}_{i,n}^{n'}\in A_{\epsilon}^{n'}(U|Y^{n'}=\widetilde{F}_{J,1}^{n'})\}.
\]
\begin{figure}[h]
\includegraphics*[draft=false,height=1.8in]{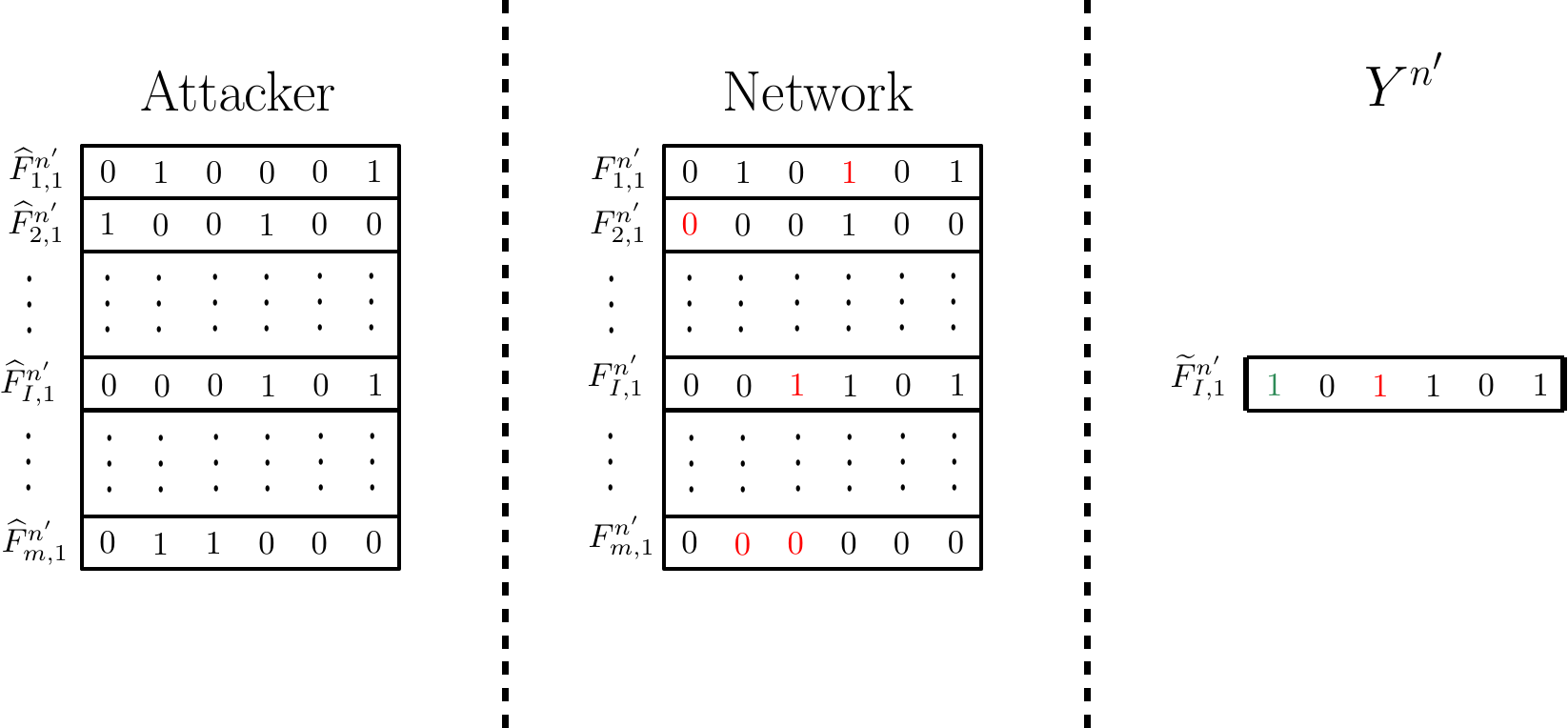}
\centering 
\caption{\textit{The figure shows the three different partial group signatures for the users. The partial signature from the network graph $g^0$ is denoted by $F_{i,1}^{n'}$. The partial signature available to the attacker in graph $g^1$ is denoted by $\widehat{F}_{i,1}^{n'}$. Finally, the partial signature received after the first step is $Y^{n'}=\widetilde{F}_{J,1}^{n'}$. The red bits are corrupted by noise in the initial phase of the attack when the attacker scans the network. The green bits are corrupted in the active phase of the attack when the attacker transmits queries to the network.} }
\label{Fig:TSS}
\end{figure}
\begin{Remark}
 When $e_i=f_i=0, i\in\{1,2\}$, we have $U^{n'}=Y^{n'}$ with probability one. So, the ambiguity set is the set of sequences $U^{n'}$ which differ with $Y^{n'}$ in at most $\epsilon n'$ elements. For small enough $\epsilon$, this set is the same as the ambiguity set in the MAP strategy.
\end{Remark}
Let $\mathcal{L}=\{i_1,i_2,\cdots,i_{|\mathcal{L}|}\}$. At time $(n'+t)$, the attacker makes a UID query about whether the user index is equal to $i_{t}$. In other words, it transmits $x_{n'+t}=r_{n+i_{t}}, t\in[1,m]$. The algorithms ends when the attacker receives the output $z_{n'+t}=1$. If the attacker fails to recover $J$ in this step, it proceeds to the second step.

In summary, the attack strategy in the first step is characterized below:
 \begin{align*}
 x_t= 
 \begin{cases}
 r_t\qquad &\text{ if } t\in [1,n'],\\
 u_{i_{t-n'}}& \text{ if } n'<t\leq n'+|\mathcal{L}|\quad \text{and} \quad y^{n'}\in \mathcal{A}_{\epsilon}^{n'}(Y).
\end{cases}
\end{align*}
The attacker fails to recover $J$ if $y^{n'}\notin \mathcal{A}_{\epsilon}^{n'}(Y)$ or $ \widehat{F}_{J,n}^{n'}\notin A_{\epsilon}^{n'}(U|Y^{n'}=\widetilde{F}_{J,1}^{n'})$. In this case, in step two, the attacker eliminates the first $n'$ groups from the graph $g^1$ and repeats the algorithm. This is repeated iteratively until the $l$th step. If the $l$th step is reached, the attacker conducts an exhaustive search by sending all possible UID queries until the de-anonymization process is complete.

\begin{Remark}
 The strategy is analogous to unstructured random coding strategies used for point-to-point channel coding. In this analogy, the set of partial group signatures $\mathcal{C}=\{\widehat{F}_{i,1}^{n'}|i\in [1,m]\}$ resembles the randomly generated codebook, the index $J$ resembles the message index, and the sequence $Y^{n'}$ resembles the sequence received by the decoder, where the codeword $\widehat{F}_{J,1}^{n'}$ is transmitted over the channel characterized by the transition probability $P_{Y|U}$.
\end{Remark}
We denote the number of queries in this attack strategy by $Q_{TSS}$. The following theorem provides bounds on the expected number of queries: 
\begin{Theorem}
For the TSS strategy, an upper bound on the expected number of queries is given by::
 \begin{align}
\mathbb{E}(Q_{TSS})= \left(n' +m2^{n'(I(U;Y)\pm \epsilon)}\right)\left(\frac{n'\epsilon^2}{n'\epsilon^2-1}\right)+\frac{m}{(n'\epsilon^2)^{l}}+O(1),
\end{align}
where $n'\leq n$.
If $n'\triangleq{=} \frac{1}{I(U;Y)+\epsilon}\log{m}$, $\epsilon\triangleq{=}n^{'-\frac{1}{3}}$ and $l\triangleq\frac{\log{m}}{\log{n'\epsilon^2}}$, we have \[\mathbb{E}(Q_{TSS})\leq  \frac{1}{I(U;Y)}\log{m}+O({\log^{\frac{2}{3}}{m}}).\] This implies
\begin{align}
\limsup_{m\to\infty}\frac{\big|\mathbb{E}(Q_{TSS})-\frac{1}{I(U;Y)}\log_2{m}\big|}{\sqrt{\log_2{m}}}<\infty.
\label{eq:th12}
\end{align}
\label{thm:3}
\end{Theorem}
\begin{proof}
 Please refer to the Appendix.
\end{proof}
\begin{Remark}
 In the noiseless case with $e_i=f_i=0, i\in \{1,2\}$, $I(U;Y)=H_b(p)$. In this case, the expected number of queries in the TSS strategy improves upon that of the MAP and GIS strategies.
\end{Remark}

\section{Conclusion}
In this paper, we have introduced new mathematical framework for the active de-anonymization problem. We have used a statistical model to capture the group membership behavior of the users in the social network as well as to model partial information available to the attacker in active de-anonymization attacks. We have investigated the minimum expected number of queries necessary for de-anonymization by analyzing the three new attack strategies devised in this paper. We have shown that in all of these strategies the expected number of queries grows logarithmically in the number of users. 

\section{Acknowledgements}
The authors are grateful to Umay Geyikci and Cezmi Mutlu for their help in formulating an earlier version of the problem studied in this paper.

\section{Appendix}
\subsection{Proof of Theorem \ref{thm:2}}
Recall that the attacker first selects   $n'$ groups randomly and uniformly. 
After finding the partial group signature corresponding to these groups. It then conducts a brute-force search over the intersection of groups which the user is a member of to de-anonymize the user. 

 We calculate the expected size of the intersection of the groups of which the user is a member. Let $r_{i_1},r_{i_2},\cdot, r_{i_{n'}}$ be the groups that are queried by the attacker. Let $L$ be defined as the size of the intersection:
\begin{align*}
 L=\Big|\bigcap_{i_k:u\in r_{i_k}}\mathcal{E}_{i_k}\Big|.
\end{align*}
We have:
\begin{align}
 &\mathbb{E}(L)=\mathbb{E}\left(\Big|\bigcap_{i_k:u\in r_{i_k}}\mathcal{E}_{i_k}\Big|\right)=\sum_{I\subset{\{i_1,i_2,\cdots,i_{n'}\}}}P\left(I=\{i_k|u\in r_k\}\right)\mathbb{E}\left(\Big|\bigcap_{{i_k:u\in r_{i_k}}}\mathcal{E}_{i_k}\Big||I=\{i_k|u\in r_k\}\right)\\
&=\sum_{I\subset\{i_1,i_2,\cdots,i_{n'}\}}P\left(I=\{i_k|u\in r_k\}\right)\mathbb{E}\left(\Big|\bigcap_{i_k\in I}\mathcal{E}_{i_k}\backslash\{u\} \Big|+1\right)
\stackrel{(a)}{=}\sum_{I\subset\{i_1,i_2,\cdots,i_{n'}\}}p^{|I|}(1-p)^{n'-|I|}\left((m-1)p^{|I|}+1\right)
 \\&=\sum_{i=1}^{n'}{n' \choose i}p^{|I|}(1-p)^{n'-|I|}\left((m-1)p^{|I|}+1\right)\nonumber
 =(m-1)\sum_{i=1}^{n'}{n' \choose i}p^{2|I|}(1-p)^{n'-|I|}+1\nonumber
\\&{=}(m-1) O(\sqrt{n'})2^{\max_{0\leq q \leq 1} {-n'( q\log_2{\frac{q}{p^2}}+(1-q)\log_2{\frac{1-q}{(1-p)}})}}+1\nonumber
{=}(m-1) O(\sqrt{n'})(1-p+p^2)^{n'}+1\label{eq:size},
\end{align}
where in (a) we have used the fact that the presence of each edge is independent of other edges. Define $c=O(\sqrt{n'})$ and $d=O(1)$.
The expected number of queries is given below:
\begin{align*}
 \mathbb{E}(Q_{GIS})
 &\stackrel{(a)}{\leq} \mathbb{E}(q|C)+\frac{1}{m^{\lambda'\log{(\frac{1}{1-p})}}}m
 = n'+(\frac{m-1}{2}) O(\sqrt{n'})(1-p+p^2)^{n'}+1
\stackrel{(b)}{=} n'+cm(1-p+p^2)^{n'}+d,
\end{align*}
where 
in (a) we have used Equation \eqref{eq:size} and the fact that $\frac{1}{p(1-p)}\log{(\frac{1}{1-p})}>1, 1>p>0$. 

Next, let 
  \begin{align}
n'=\left(\frac{1}{p(1-p)}+\frac{1}{\log_2{\frac{1}{1-p}}}\right)\log_2{m}.
  \end{align}
Then,
\begin{align*}
 \mathbb{E}(Q_{GIS})&=n'+cm(1-p+p^2)^{n'}+d
 \stackrel{(a)}{\leq} n'+cm2^{-n'p(1-p)}+d
  =n'+cm2^{-p(1-p)(\lambda'+\frac{1}{\log_2{\frac{1}{1-p}}})\log_2{m}}+d
 \\&=n'+c2^{-\frac{p(1-p)}{\log_2{\frac{1}{1-p}}}\log_2{m}}+d
 =n'+O(1),
\end{align*}
where in (a) we have used the inequality $(1-y)^n\leq 2^{-ny}, y\in (0,1), n\in \mathbb{N}$.

  \subsection{Proof of Theorem \ref{thm:1}}
Without loss of generality
Assume that user $u_J$ is to be de-anonymized. As explained above, initially, the user's partial class signature is determined in the first $n'$ queries.
Hence, after $n'$ queries, the attacker has access to the $n'$-length binary vector $\underline{F}'$, where $F'_k=F_{J,k}, k\in [1,n']$. If there exists a unique user for which  $\underline{F}'={F}_{j,1}^{n'}$, then $J=j$ and the algorithm stops (this is true since such a user would have the maximum posterior probability among all users, so the corresponding UID query will be made first and the user is de-anonymized). Otherwise, if there exists more than one user with partial class signature equal to $\underline{F}'$, then, the attacker proceeds to step two which involves transmitting UID queries corresponding to the maximum posteriori probabilities. Due to symmetry, the maximum posteriori search over the set of users with partial signature $\underline{F}'$ is equivalent to an exhaustive search over these users. So, the attacker forms an ambiguity set which contains all users which have the same partial class signatures. Then, it transmits queries corresponding to the users in the ambiguity set.

Without loss of generality, let us assume that $I=1$ (i.e. the first user is to be de-anonymized).  

 Formally, The attacker defines the following ambiguity set: 
\begin{align*}
 \mathcal{L}\triangleq\{u_i|\underline{F}'={F}_{i,1}^{n'}\}.
\end{align*}
Let $\mathcal{L}$ be the set of users with non-zero MAP. Let the indices of users in the ambiguity set $\mathcal{L}$ be denoted by $\{{(i_1)}, {(i_2)},\cdots,i_{(|\mathcal{L}|)}\}$. The transmitted symbol at the $t$th step is $x_t=u_{l_{(t)}-n'}, t\in [n'+1,n'+|\mathcal{L}|]$.

 Clearly, $\mathbb{E}(q)= n'+\frac{\mathbb{E}(|\mathcal{L}|)}{2}$, where the first $n'$ queries pertain to determining the partial class signature $\underline{F}'$. For the second term $\frac{\mathbb{E}(|\mathcal{L}|)}{2}$,
 we have $\mathbb{E}(|\mathcal{L}|)= \sum_{i=1}^m P(J\neq i)P(u_i\in \mathcal{L}|J\neq i)+1.$
Without loss of generality assume that $J=1$ and $i\neq 1$. Then, by similar argument as in Theorem \ref{thm:2}, $P(E|J\neq i,J=1)=O(\sqrt{n'})(p^2+(1-p)^2)^{n'}$.
Let $E$ be the event that $u_i\in \mathcal{L}$, then:
\begin{align}
 \nonumber P(E|J\neq i,J=1)&=\sum_{\underline{F}'\subset \{0,1\}^{n'}} P(\underline{F}'={F}_{1,1}^{n'}={F}_{i,1}^{n'})= \sum_{\underline{F}'\subset \{0,1\}^{n'}}  P(\underline{F}'={F}_{1,1}^{n'})P(\underline{F}'={F}_{i,1}^{n'})\nonumber\\
 &= \sum_{\underline{F}'\subset \{0,1\}^{n'}} p^{2w_H(\underline{F}')}(1-p)^{2(n-w_H(\underline{F}'))}\nonumber= \sum_{i=1}^{n'} {n' \choose i} p^{2i}(1-p)^{2(n'-i)}\nonumber\\
 &\stackrel{(a)}{=} \sum_{i=1}^{n'} \frac{1}{\sqrt{\lambda n'}}2^{nH_b(\frac{i}{n'})}\left(1+O(\frac{1}{n'})\right) p^{2i}(1-p)^{2(n'-i)}\nonumber,
 \end{align}
 
 where in (a) we have used the following equality 
\begin{align*}
 {n' \choose i}=\frac{1}{\sqrt{\lambda n'}}2^{nH_b(\frac{i}{n'})}\left(1+O(\frac{1}{n'})\right),
\end{align*}
where $\lambda=2\pi p(1-p)$ and $H_b$ is the binary entropy function. The proof of the equality follows from Sterling's approximation $n'!=\sqrt{2\pi n} e^{-n'}{n'}^{n'}\left(1+O(\frac{1}{n'})\right)$. So, we have:
\begin{align}
& \nonumber P(E)\stackrel{(b)}\leq O(n') \max_{0\leq q \leq 1}    \frac{1}{\sqrt{\lambda n'}}2^{n'H_b(q)}\left(1+O(\frac{1}{n'})\right) p^{2n'q}(1-p)^{2n'(1-q)}\nonumber\\
 &=O(\sqrt{n'})2^{\max_{0\leq q \leq 1}   {n'H_b(q)+ {2n'q}\log_2{p}+{2n'(1-q)}\log_2{(1-p)}}}\nonumber\\
 &=O(\sqrt{n'})2^{\max_{0\leq q \leq 1}   {n'(H_b(q)+ {q}\log_2{p^2}+{(1-q)}\log_2{(1-p)^2})}}\nonumber\\
 &=O(\sqrt{n'})2^{\max_{0\leq q \leq 1} {-n'( q\log_2{\frac{q}{p^2}}+(1-q)\log_2{\frac{1-q}{(1-p)^2}})}}\nonumber\\
 &=O(\sqrt{n'})2^{\max_{0\leq q \leq 1} {-n'\left( q\log_2{\frac{q}{\frac{p^2}{p^2+(1-p)^2}}+(1-q)\log_2{\frac{1-q}{\frac{(1-p)^2}{p^2+(1-p)^2}}}}-\log_2{(p^2+(1-p)^2)}\right)}}\nonumber\\
 &=O(\sqrt{n'})2^{\max_{0\leq q \leq 1} {-n'D(q||\frac{p^2}{p^2+(1-p)^2})+n'\log_2{(p^2+(1-p)^2)}}}\nonumber
 \\&\stackrel{(c)}{\leq} O(\sqrt{n'})2^{n'\log_2{(p^2+(1-p)^2)}}\nonumber
 =O(\sqrt{n'})(p^2+(1-p)^2)^{n'},
\end{align}
where $D(\cdot||\cdot)$ is the Kullback-Leibler Divergence. To show (b), let us define $t_i\triangleq \frac{1}{\sqrt{\lambda n'}}2^{n'H_b(\frac{i}{n'})}\left(1+O(\frac{1}{n'})\right) p^{2i}(1-p)^{2(n'-i)}$. Then,
\begin{align}
&\max_{i\in [1,n']}t_i \leq \sum_{i=1}^{n'}t_i\leq n'\max_{i\in [1,n']} t_i
 \Rightarrow 1\leq \frac{\sum_{i=1}^{n'}t_i}{\max_{i\in [1,n']}t_i}\leq n'
 \Rightarrow \sum_{i=1}^{n'}t_i= O(n')\max_{i\in [1,n']}t_i \leq O(n')\max_{0\leq q\leq 1}t_{qn'}.\label{eq:note}
\end{align}
%
%
%
Also, (c) follows from the fact that divergence is positive. So, as $n\to\infty$, we have $\mathbb{E}(|\mathcal{L}|)\to cm(p^2+(1-p)^2)^{n'}+1$, where $c=O({\sqrt{n'}})$. This completes the proof of \eqref{eq:th11}
. To prove \eqref{eq:th12}, let $\lambda=2p(1-p)$ and $n'= \frac{1}{\lambda}\log_2{m}$, then from \eqref{eq:th11} we have:
\begin{align}
 \mathbb{E}(Q_{MAP})&\leq n'+\frac{cm}{2}(p^2+(1-p)^2)^{n'}\label{eq:bound1}
=n'+\frac{cm}{2}(1-2p(1-p))^{n'}
 \\&\stackrel{(a)}{\leq} n'+\frac{cm}{2}2^{-\lambda n'}
 =n'+\frac{c}{2}2^{-\lambda n'+\log_2{m}}\label{eq:bound1.5}
 \\&=\frac{1}{\lambda}log{m}+\frac{c}{2}2^{\lambda\frac{1}{\lambda}log{m}-\log_2{m}}
 \stackrel{(b)}{=}\frac{1}{\lambda}log{m}+O(\sqrt{\log_2{m}}),\nonumber
 \end{align}
where in (a) we have used the inequality $(1-y)^n\leq 2^{-ny}, y\in (0,1), n\in \mathbb{N}$, and in (b) we have used $c=O(\sqrt{n'})$. 
This completes the proof.

\subsection{Proof of Theorem \ref{thm:3}}

 Fix $n',l\in \mathbb{N}$ and $\epsilon>0$. After $n'$ queries, the attacker has access to the $n'$-length binary vector $\widetilde{\underline{F}'}$, where $\widetilde{F}'_k=\widetilde{F}_{J,k}, k\in [1,n']$. The attacker checks the $\epsilon$-typicality of the vector $\widetilde{\underline{F}'}$ with respect to the distribution $P_{Y}(\cdot)$. Let $B$ be the event that $\widetilde{\underline{F}'}\in A_{\epsilon}^{n'}(Y)$. If $\widetilde{\underline{F}'}\in A_{\epsilon}^{n'}(Y)$. Then the attacker finds all partial group signatures $\widehat{F}_{i,1}^{n'}, i\in [1,m]$ which are conditionally typical with respect to $P_{U|Y}$ given the sequence $Y^{n'}= \widetilde{\underline{F}'}$. More precisely, it forms the following ambiguity set,
 $\mathcal{L}=\{i| \widehat{F}_{i,n}^{n'}\in A_{\epsilon}^{n'}(U|Y^{n'}=\widetilde{F}_{J,1}^{n'})\}$.
Denote the indices of users in the ambiguity set $\mathcal{L}$ by $\{{(i_1)}, {(i_2)},\cdots,i_{(|\mathcal{L}|)}\}$. The transmitted symbol at the $t$th step is $x_t=u_{l_{(t)}-n'}, t\in [n'+1,n'+|\mathcal{L}|]$. Let $C$ be the event that $\widehat{F}_{J,1}^{n'}\in A_{\epsilon}^{n'}(U|Y^{n'}=\widetilde{F}_{J,1}^{n'})$. If the event $B\cap C$ occurs then the index $J$ is de-annonymized. Otherwise, the attacker removes the first $n'$ groups from the network and repeats the TSS for the new network with $m$ users and $n-n'$ groups. Let $Q_{TSS}^{n,m}$ denote the number of queries in a TSS for active attacks on a network with $n$ groups and $m$ users. Then,
\begin{align*}
&\mathbb{E}(Q_{TSS}^{n,m})
\\&= 
 \mathbb{E}(Q_{TSS}^{n,m}|B\cap C)P(B\cap C)+  \mathbb{E}(Q_{TSS}^{n,m}|B^c\cup C^c)P(B^c\cup C^c) 
\\ &\leq\mathbb{E}(Q_{TSS}^{n,m}|B\cap C)+  \mathbb{E}(Q_{TSS}^{n-n',m})P(B^c\cup C^c)
\\&\stackrel{(a)}{\leq} n'+\mathbb{E}(\big|\mathcal{L}\big|\Big| B) +\mathbb{E}(Q_{TSS}^{n-n',m})P(B^c\cup C^c),
\end{align*}
where $(a)$ follows from the fact that if $B\cap C$ occurs, then the algorithm ends in at most  $n'+\big|\mathcal{L}\big|$ steps, and that the size of $\mathcal{L}$ is independent of the event $F$ given event $E$. First, we derive bounds on $\mathbb{E}(\big|\mathcal{L}\big|\Big| B) $:
\begin{align*}
& \mathbb{E}(\big|\mathcal{L}\big|\Big| B)= \sum_{i=1}^m P(u_i\in \mathcal{L}\Big|B)= \mathbb{E}\left(\big|\{i| \widehat{F}_{i,n}^{n'}\in A_{\epsilon}^{n'}(U|Y^{n'}=\widetilde{F}_{J,1}^{n'})\}\big|\Big| B\right)
\\& =\sum_{y^{n'}\in A_{\epsilon}^{n'}(Y)} P(Y^{n'}=y^{n'}|B) \mathbb{E}\left(\big|\{i| \widehat{F}_{i,n}^{n'}\in A_{\epsilon}^{n'}(U|y^{n'})\}\big|\Big|y^{n'}\right)
\\& =\sum_{y^{n'}\in A_{\epsilon}^{n'}(Y)} P(Y^{n'}=y^{n'}|B) \mathbb{E}\left(\sum_{i=1}^m \mathbbm{1}\left(\widehat{F}_{i,n}^{n'}\in A_{\epsilon}^{n'}(U|y^{n'})\right)\big|y^{n'}\right)
\end{align*}
\begin{align*}
\\&=\sum_{y^{n'}\in A_{\epsilon}^{n'}(Y)}  \sum_{i=1}^mP(Y^{n'}=y^{n'}|B) P\left(\widehat{F}_{i,n}^{n'}\in A_{\epsilon}^{n'}(U|y^{n'})\big|y^{n'}\right)
\\&=m\sum_{y^{n'}\in A_{\epsilon}^{n'}(Y)} P(Y^{n'}=y^{n'}|B) P\left(\widehat{F}_{1,n}^{n'}\in A_{\epsilon}^{n'}(U|y^{n'})\big|y^{n'}\right)
\\&= m\sum_{y^{n'}\in A_{\epsilon}^{n'}(Y)} P(Y^{n'}=y^{n'}|B) \sum_{u^{n'} \in A_{\epsilon}^{n'}(U|y^{n'})}P\left(\widehat{F}_{1,n}^{n'}=u^{n'}|y^{n'}\right)
\\&=m \sum_{(u^{n'},y^{n'})\in A_{\epsilon}^{n'}(U,Y)} P(Y^{n'}=y^{n'}|B)P\left(\widehat{F}_{1,n}^{n'}=u^{n'}|y^{n'}\right)
\\&\stackrel{.}{=} m2^{n'(H(U,Y)\pm\epsilon)} 2^{-n'(H(U)\pm \epsilon)}2^{-n'(H(Y)\pm \epsilon)}
\\&\stackrel{.}{=} m2^{n'(I(U;Y)\pm \epsilon)},
\end{align*}
where for a function $f(\cdot)$ and $\epsilon, x,y\in \mathbb{R}$, we write $y\stackrel{.}{=}f(x\pm \epsilon)$ to denote $y\in [f(x-\epsilon),f(x+\epsilon)]$, and in the last two inequalities we have used Lemma \ref{lem:typ2}.
{Furthermore, $P(E\cap F)\geq 1-\frac{1}{n'\epsilon^2}$ from Lemma \ref{lem:typ1}}. As a result, 
\begin{align*}
 \mathbb{E}(Q_{TSS}^{n,m})&\leq n' +m2^{n'(I(U;Y)\pm \epsilon)}+ \frac{1}{n'\epsilon^2}\mathbb{E}(Q_{TSS}^{n-n',m}).
\end{align*}
 The attacker repeats this algorithm iteratively for $l$ consecutive steps. Hence, 
 \begin{align*}
 &\mathbb{E}(Q_{TSS}^{n,m})
 \\&\leq n' +m2^{n'(I(U;Y)\pm \epsilon)}+ \frac{1}{n'\epsilon^2}\left(n'+m2^{n'(I(U;Y)\pm \epsilon)}+ \frac{\mathbb{E}(Q_{TSS}^{n-2n',m})}{n'\epsilon^2}\right)\\
 &\leq \!\left(n' \!\!+\!m2^{n'(I(U;Y)\pm \epsilon)}\right)\left(1+ \frac{1}{n'\epsilon^2}+ \frac{1}{(n'\epsilon^2)^2}\cdots  \frac{1}{(n'\epsilon^2)^l}\right)\!+\! \frac{m}{(n'\epsilon^2)^{l+1}}\\
 &\leq  \left(n' +m2^{n'(I(U;Y)\pm \epsilon)}\right)\left(\frac{n'\epsilon^2}{n'\epsilon^2-1}\right)+\frac{m}{(n'\epsilon^2)^{l}}.
\end{align*}
For $n'{=} \frac{1}{I(U;Y)+\epsilon}\log{m}$, $\epsilon=n^{'-\frac{1}{3}}$ and $l=\frac{\log{m}}{\log{n'\epsilon^2}}$, we have
\begin{align*}
  \mathbb{E}(Q_{TSS}^{n,m})\leq  \frac{1}{I(U;Y)}\log{m}+O({\log^{\frac{2}{3}}{m}}).
\end{align*}
This completes the proof.

\bibliographystyle{unsrt}

\begin{thebibliography}{1}

\bibitem{Beyah}
S.~Ji, W.~Li, N.~Z. Gong, P.~Mittal, and R.~Beyah.
\newblock Seed-based de-anonymizability quantification of social networks.
\newblock {\em IEEE Transactions on Information Forensics and Security},
  11(7):1398--1411, July 2016.

\bibitem{kruegel}
G.~Wondracek, T.~Holz, E.~Kirda, and C.~Kruegel.
\newblock A practical attack to de-anonymize social network users.
\newblock In {\em 2010 IEEE Symposium on Security and Privacy}, pages 223--238,
  May 2010.

\bibitem{ref1}
S.~Parthasarathi G. Friedland R. Sommer R.~Teixeira. O.~Goga, H.~Lei.
\newblock Exploiting innocuous activity for correlating users across sites.
\newblock {\em Proceedings of the 22Nd International Conference on World Wide
  Web}, pages 447--458, 2013.

\bibitem{ref2}
Arvind Narayanan and Vitaly Shmatikov.
\newblock De-anonymizing social networks.
\newblock In {\em Proceedings of the 2009 30th IEEE Symposium on Security and
  Privacy}, SP '09, pages 173--187, Washington, DC, USA, 2009. 

\bibitem{ref3}
Mudhakar Srivatsa and Mike Hicks.
\newblock Deanonymizing mobility traces: Using social network as a
  side-channel.
\newblock In {\em Proceedings of the 2012 ACM Conference on Computer and
  Communications Security}, CCS '12, pages 628--637, New York, NY, USA, 2012.
 

\bibitem{ped}
Pedram Pedarsani and Matthias Grossglauser.
\newblock On the privacy of anonymized networks.
\newblock In {\em Proceedings of the 17th ACM SIGKDD International Conference
  on Knowledge Discovery and Data Mining}, KDD '11, pages 1235--1243, New York,
  NY, USA, 2011. ACM.

\bibitem{csiszarbook}
I.~Csisz\'{a}r and J.~Korner.
\newblock {\em Information Theory: Coding Theorems for Discrete Memoryless
  Systems}.
\newblock Academic Press Inc. Ltd., 1981.

\end{thebibliography}

\end{document}